\documentclass[envcountsect, 11pt, fleqn]{llncs}

\usepackage[left=1in,right=1in,top=1in,bottom=1in]{geometry}
\usepackage[numbers]{natbib}
\usepackage{subfigure}
\usepackage{amsmath}
\usepackage{amsfonts}
\usepackage{bm}
\usepackage{mathdots}
\usepackage[ruled]{algorithm2e}
\usepackage{varwidth}
\usepackage{algpseudocode}
\usepackage{mathtools}
\usepackage{enumerate}
\usepackage{graphicx}
\usepackage{amssymb}

\usepackage[citecolor=blue, linkcolor=blue, colorlinks=true]{hyperref}

\renewcommand{\tabcolsep}{0.5em}

\newcommand\refTheorem[1]{Theorem~\hyperref[#1]{\ref*{#1}}}
\newcommand\refLemma[1]{Lemma~\hyperref[#1]{\ref*{#1}}}
\newcommand\refCorollary[1]{Corollary~\hyperref[#1]{\ref*{#1}}}
\newcommand\refProposition[1]{Proposition~\hyperref[#1]{\ref*{#1}}}
\newcommand\refRemark[1]{Remark~\hyperref[#1]{\ref*{#1}}}
\newcommand\refFigure[1]{Figure~\hyperref[#1]{\ref*{#1}}}
\newcommand\refAlgorithm[1]{Algorithm~\hyperref[#1]{\ref*{#1}}}
\newcommand\refEquation[1]{(\hyperref[#1]{\ref*{#1}})}
\newcommand\refExample[1]{Example~\hyperref[#1]{\ref*{#1}}}
\newcommand\refSection[1]{Section~\hyperref[#1]{\ref*{#1}}}
\newcommand\refDefinition[1]{Definition~\hyperref[#1]{\ref*{#1}}}

\newcommand\prob[1]{\textsf{Pr}\left[#1\right]}
\newcommand\probsub[2]{\textsf{Pr}_{#1}\left[#2\right]}

\newcommand\expected[1]{\mathbb{E}\left[#1\right]}
\newcommand\expectedsub[2]{\mathbb{E}_{#1}\left[#2\right]}
\newcommand\expnobrack{\mathbb{E}}

\newcommand\indicator[1]{\mathbb{I}\left[#1\right]}

\newcommand\SW{\textsf{SW}}

\newcommand\all{\mathrm{all}}

\newcommand\OPT{\textsf{OPT}}
\newcommand\ALG{\textsf{ALG}}

\begin{document}
\title{Approximately Efficient Two-Sided Combinatorial Auctions}












\author{Riccardo Colini-Baldeschi \and Paul Goldberg \and Bart de Keijzer \and Stefano Leonardi  \and Tim Roughgarden \and Stefano Turchetta}


\institute{LUISS Rome, \email{rcolini@luiss.it} \and University of Oxford, \email{paul.goldberg@cs.ox.ac.uk}, \email{stefano.turchetta@cs.ox.ac.uk} \and Centrum Wiskunde \& Informatica (CWI), Amsterdam, \email{keijzer@cwi.nl} \and Sapienza University of Rome, \email{leonardi@diag.uniroma1.it} \and Stanford University, \email{tim@cs.stanford.edu}}



\maketitle
\thispagestyle{empty}
\setcounter{page}{0}

\begin{abstract}
Mechanism design for one-sided markets has been investigated for several decades in economics and  in computer science.  More recently, there has been an increased attention on mechanisms for two-sided markets, in which buyers and sellers act strategically.  For two-sided markets, an impossibility result of Myerson and Satterthwaite states that no mechanism can simultaneously satisfy {\em individual rationality (IR)}, {\em incentive compatibility (IC)}, {\em strong budget-balance (SBB)}, and be efficient. Follow-up work mostly focused on designing mechanisms that trade off among these properties, in settings where there exists one single type of items for sale, the buyers ask for one unit of the items, and the sellers initially hold one unit of the item.

On the other hand, important applications to web advertisement, stock exchange, and frequency spectrum allocation, require us to consider two-sided combinatorial auctions in which buyers have preferences on subsets of items, and sellers may offer multiple heterogeneous items. No efficient mechanism was known so far for such two-sided combinatorial markets. This work provides the first IR, IC and SBB mechanisms that provides an $O(1)$-approximation to the optimal social welfare for two-sided markets. 
An initial construction yields such a mechanism, but exposes a conceptual problem in the traditional SBB notion. This leads us to define the stronger notion of \emph{direct trade strong budget balance (DSBB)}. 

We then proceed to design mechanisms that are IR, IC, DSBB, and again provide an $O(1)$-approximation to the optimal social welfare. 
Our mechanisms work for any number of buyers with XOS valuations  -- a class in between submodular and subadditive functions -- and any number of sellers. We provide a mechanism that is dominant strategy incentive compatible (DSIC) if the sellers each have one item for sale, and one that is bayesian incentive compatible (BIC) if sellers hold multiple items and have additive valuations over them. Finally, we present a DSIC mechanism for the case that the valuation functions of all buyers and sellers are additive. We prove our main results by showing that there exists a variant of a sequential posted price mechanism, generalised to two-sided combinatorial markets, that achieves the desired goals.
\end{abstract}

\section*{Acknowledgement}
Partially supported by the ERC Advanced Grant 788893 AMDROMA ``Algorithmic and Mechanism Design Research in Online Markets'' and MIUR PRIN project ALGADIMAR ``Algorithms, Games, and Digital Markets''.

\newpage

\section{Introduction}
One-sided markets have been studied in economics for several decades and more recently in computer science.
Mechanism Design in one-sided markets aims to find an efficient (high-welfare) allocation of a set of items to a set of agents, while ensuring that truthfully reporting the input data is the best strategy for the agents. The cornerstone method in mechanism design is the  Vickrey-Clarke-Groves (VCG) mechanism \cite{vickrey61, clarke71, groves73} that optimizes the social welfare of the agents while providing the right incentives for truth-telling: VCG mechanisms are \emph{dominant strategy incentive compatible (DSIC)}, and in many mechanism design settings, VCG is also {\em individually rational (IR)}. IR requires that participating in the mechanism is beneficial to each agent. DSIC requires that truthfully reporting one's preferences to the mechanism is a dominant strategy for each agent, independently of what the other agents report. 

Recently, increased attention has been on the problems that arise in two-sided markets, in which the set of agents is partitioned into \emph{buyers} and \emph{sellers}. As opposed to the one-sided setting (where one could say that the mechanism itself initially holds the items), in the two-sided setting the items are initially held by the set of sellers, who express valuations over the items they hold, and who are assumed to act rationally and strategically. The mechanism's task is now to decide which buyers and sellers should trade, and at which prices. There is a growing interest in two-sided markets that can be attributed to various important applications. Examples range from selling display-ads on ad exchange platforms, the US FCC spectrum license reallocation, and stock exchanges. Two-sided markets are usually studied in a Bayesian setting: there is public knowledge of probability distributions, one for each buyer and one for each seller, from which the valuations of the buyers and sellers are drawn.

In two-sided markets, a further important requirement is {\em strong budget-balance (SBB)}, which states that monetary transfers happen only among the agents in the market, i.e., the buyers and the sellers are allowed to trade without leaving to the mechanism any share of the payments, and without the mechanism adding money into the market. 
A weaker version of SBB often considered in the literature is \emph{weak budget-balance (WBB)}, which only requires the mechanism not to inject money into the market.  

Unfortunately, Myerson and Satterthwaite \cite{ms83} proved that it is impossible for an IR, Bayesian incentive compatible (BIC),\footnote{Bayesian incentive compatibility is a form of incentive compatibility that is less restrictive than DSIC. It only requires that reporting truthfully is in expectation the best strategy for a player when everyone else also does so, i.e., truthful reporting is a Bayes-Nash equilibrium.}
and WBB mechanism to maximize social welfare in such a market, even in the \emph{bilateral trade} setting, i.e., when there is just one seller and one buyer.\footnote{The VCG mechanism can also be applied to two-sided markets; however, in this setting, VCG is either not IR or it does not satisfy WBB.}



Despite the  numerous above mentioned practical contexts that need the application of \emph{combinatorial} two-sided  market mechanisms, we are not aware of any mechanism that approximates the social welfare while meeting the IR, IC and SBB requirements. 
The purpose of this paper is to provide mechanisms that satisfy these requirements and achieve an $O(1)$-approximation to the social welfare for a broad class of agents' valuation functions. We do, in fact, design mechanisms that work under the assumption of the valuations being {\em fractionally subadditive (XOS)}, a generalisation of submodular functions that are contained in the class of subadditive functions. 

Our work builds upon previous work on an important special case of a two-sided market that is the one in which each seller holds a single item, items are identical, and each agent is only interested in holding a single item. In this setting, the valuations of the agents are thus given by a single number, representing the agent's appreciation for holding an item. A mechanism for this setting is known in the literature as a \emph{double auction}.
The goal of several works on double auctions \cite{mcafee92, SW89, sw02} has been that of trading off the achievable social welfare with the strength of the incentive compatibility and budget balance constraints. 
A recent work addressed the problem of approximating social welfare in double auctions and related problems under the WBB requirement. D\"{u}tting, Talgam-Cohen, and Roughgarden \cite{dtr14} indeed proposed a greedy strategy that combines the one-sided VCG mechanism, independently applied to buyers and to sellers with the trade-reduction mechanism of McAfee \cite{mcafee92}. They obtain IR, DSIC, WBB mechanisms with a good approximation of the social welfare, for knapsack, matching and matroid allocation constraints.  
More recently, Colini-Baldeschi et al.~\cite{cdlt16} presented the first double auction that satisfies IR, DSIC, and SBB, and approximates the optimal (expected) social welfare up to a constant factor. These results hold for any number of buyers and sellers with arbitrary, independent distributions on valuations. The mechanisms are also extended to the setting where there is an additional matroid constraint on the set of buyers who can purchase an item.

\subsection{The Model}
As stated above, the set of agents is partitioned into a set of {\em sellers}, each of which is initially endowed with a set of heterogeneous items, and a set of {\em buyers}, having no items initially. Buyers have money that can be used to pay for items. Every agent has its own, private valuation function, which maps subsets of the items to numbers, and agents are assumed to optimize their (quasi-linear) utility, which is given by the valuation of the set of items that the mechanism allocates to the agents, minus the payment that the mechanism asks from the agents. A seller will typically receive money instead of pay money, which we model by a negative payment. 

For each agent we are given a (publicly known) probability distribution over a set of valuation functions, from which we assume her valuation function is drawn. The mechanism and the other agents have no knowledge of the actual valuation function of the $i$-th agent, but only of her probability distribution. The general aim of the mechanism is to reallocate the items so as to maximize the expected social welfare (the sum of the agents' valuations of the resulting allocation). 

Let \OPT\ be the expected social welfare of an optimal allocation of the items. Note that this is a well-defined quantity, even though computing an optimal allocation may be computationally hard, and even though there might not exist an appropriate mechanism that is guaranteed to always output an optimal allocation. 

We are interested in mechanisms that satisfy \emph{individual rationality (IR)}, \emph{dominant strategy incentive compatibility (DSIC)} (or failing that, the weaker notion of \emph{Bayesian incentive compatibility (BIC)}), and \emph{strong budget-balance (SBB)}, and that reallocate the items in such a way that the expected social welfare is within some constant fraction of \OPT, where expectation is taken over the given probability distributions of the agents' valuations, and over the randomness of the allocation that the mechanism outputs. The obstacle of interest to us is not the computational one, but the requirement of achieving mechanisms with such approximation guarantees, that are also DSIC, IR, and SBB. 
The main focus of this paper is the existence of mechanisms that obtain a constant approximation to the optimal social welfare and not the computational issues related to them.
We remark, however, that our mechanisms can be implemented in polynomial time at the cost of an additional welfare loss of a factor $c$, where $c$ is the best-known approximation factor for the problem of optimising buyers' social welfare, if both a poly-time approximation algorithm and an approximation for the query oracle exist (see e.g., \cite{fgl15}).


\subsection{Overview of the Results}

The present paper starts off by showing that there is a straightforward trick that one may apply to turn any WBB mechanism into an SBB one, with a small loss in approximation factor. The technique is to pre-select one random agent that is taken out of the market a priori whose role is to receive all leftover money, and it can be applied in order to obtain from \cite{bd14} a $O(1)$-approximate DSIC, IR, and SBB mechanism for a very broad class of markets. This therefore demonstrates a weakness in the current notion of SBB, which motivates the introduction of the strengthened notion \emph{direct trade strong budget balance (DSBB)}, that requires that a monetary transfer between two agents in the market is only possible in case the agents trade items. 



The goal of our proposed mechanisms is twofold: (i) achieve a constant approximation to the optimal social welfare, and (ii) design mechanisms that respect the stronger notion of DSBB. Note that with non-unit-supply sellers, a constant approximation is not known even in the context of WBB or standard SBB.\footnote{The mechanism proposed in \cite{bd14} achieves a constant approximation if the size of the initial endowment of each agent is bounded by a constant. Otherwise the mechanism achieves a logarithmic approximation to the optimal social welfare.}
We present the following mechanisms:
\begin{itemize}
\item A $6$-approximate DSIC mechanism for buyers with XOS-valuations and sellers with one item at their disposal (i.e., \emph{unit-supply sellers});
\item a $6$-approximate BIC mechanism for buyers with XOS-valuations and non-unit supply sellers with additive valuations;
\item a $6$-approximate DSIC mechanism for buyers with additive valuations and sellers with additive valuations.
\end{itemize} 
XOS functions lie in between the class of monotone submodular functions and subadditive functions in terms of generality.
Additive and XOS functions are frequently used in the mechanism design literature. To our knowledge, these are the first mechanisms for these two-sided market settings that are simultaneously incentive compatible, (D)SBB, IR, and approximate the optimal social welfare to within a constant factor.

A first ingredient needed to obtain our results is the extension of {\em two-sided sequential posted price mechanisms (SPMs)} \citep{cdlt16} for double auctions to two-sided markets. SPMs are a particularly elegant and well-studied class of mechanisms for one-sided markets. All the presented SPMs do not require any assumption on the arrival order of the agents.
A second ingredient of our result is to use the \emph{expected marginal contribution of an item to the social welfare} as the price of the item in a sequential posted price mechanism for buyers with XOS valuations \cite{fgl15}. 

\subsection{Related Work}\label{sec:related}
Due to the impossibility result of \citep{ms83}, no two-sided mechanism can simultaneously satisfy BIC, IR, WBB and be socially efficient, even in the simple bilateral trade setting.
Follow-up work thus had to focus on designing mechanisms that trade off among these properties.

Some paper of the Economics literature studied the convergence rate to social efficiency as a function of the number of agents when all sellers' and buyers' valuations are independently respectively drawn from identical regular distributions, while satisfying IR and WBB. \citet{gs89} showed that duplicating the number of agents by $\tau$ results in a market where the optimal IR, IC, WBB mechanism's inefficiency goes down as a function of $O(\log{\tau} / \tau^2)$. \citep{rsw94, sw02} investigated a family of non-truthful double auction mechanisms, parameterized by a value $c \in [0,1]$. We remark that the results mentioned above only hold for unit-demand buyers and unit-supply sellers, identical valuation distributions, and the hidden constants in these asymptotic result depend on the particular distribution. In contrast, our interest is in finding \emph{universal} constant approximation guarantees for combinatorial settings.

In \citet{mcafee92}, an IC, WBB, IR double auction is proposed that extracts at least a $(1 - 1/\ell)$ fraction of the maximum social welfare, where $\ell$ is the number of traders in the optimal solution. 

Optimal revenue-maximizing Bayesian auctions were characterized in \citep{myerson81}, which provides an elegant tool applicable to single-parameter, one-sided auctions. Various subsequent articles
dealt with extending these results. Related to our work is \citep{dgtz14}, which studied maximizing the auctioneer's revenue in Bayesian double auctions. The same objective was studied in \citep{dghk02} yet in the \emph{prior-free} model.
Recently, \citet{dtr14} provided black-box reductions from WBB double auctions to one-sided mechanisms. They are for a prior-free setting and can be applied with matroid, knapsack, and matching feasibility constraints on the allocations. More recently, \cite{cdlt16} presented the first IR, IC and SBB mechanism for double auctions that $O(1)$-approximates the optimal social welfare. 
In \cite{DBLP:journals/corr/Segal-HaleviHA16}, mechanisms for some special cases of two-sided markets are presented that work by a combination of random sampling and random serial dictatorship.
The mechanism is IR, SBB and DSIC and its \emph{gain from trade} approaches the optimum when the market is sufficiently large.
Mechanisms that are IC, IR, and SBB have been given for bilateral trade in \citep{bd14}. In addition to this, the authors proposed a WBB mechanism for a general class of combinatorial exchange markets. We will use this result to construct our initial mechanism. 

Sequential posted price mechanisms (SPMs) in one-sided markets have been introduced in \citep{sandholm2004} and have since gained attention due to their simplicity, robustness to collusion, and their easy
implementability in practical applications. One of the first theoretical results concerning SPMs is an asymptotic comparison among three different types of single-parameter mechanisms \citep{bh08}. They were later studied for the objective of revenue maximisation in \citep{chms10}. Additionally, \cite{kw12, dk15} strengthen these results further. 
Very relevant to our work is the paper of Feldman et al. \cite{fgl15} showing that sequential posted price mechanisms can approximate social welfare up to a constant factor of $1/2$ for XOS valuation functions if the published price for an item is equal to the expected additive contribution of the item to the social welfare.

\section{Preliminaries}\label{sec:prelims}
As a general convention, we use boldface notation for vectors and use $[a]$ to denote the set $\{1, \ldots, a\}$. We will use $\mathbb{I}(X)$ to denote the indicator function that maps to $1$ if and only if event/fact $X$ holds. 
\smallbreak
\noindent \textbf{Markets.}
A \emph{two-sided market} comprises a set of two distinct types of agents: the \emph{sellers}, who initially hold items for sale, and the \emph{buyers}, who are interested in buying the sellers' items. All agents possess a monotone and normalized valuation function, mapping subsets of items to $\mathbb{R}_{\geq 0}$. \footnote{By \emph{monotone} we mean that getting more items cannot decrease an agent's overall valuation, and with normalized we mean that the empty set is mapped to $0$.} Formally, we represent a two-sided market as a tuple $(n, m, k,\bm{I},\bm{G},\bm{F})$,  where $[n]$ denotes the set of buyers, $[m]$ denotes the set of sellers, $[k]$ denotes the set of all items for sale, $\bm{I} := (I_1, \ldots, I_m)$ is a vector of (mutually disjoint) sets of items initially held by each seller, called the \emph{initial endowment}, where it holds that $\bigcup_{j=1}^m I_j = [k]$. Vectors $\bm{G} = (G_1, \ldots, G_n)$ and $\bm{F} = (F_1, \ldots, F_m)$ are vectors of probability distributions, from which the buyers' and sellers' valuation functions are assumed to be drawn. 
A \emph{(combinatorial) exchange market} is a more general version of the above defined two-sided market where an agent can act as both a buyer and a seller. Thus, everyone may initially own items and may both sell and buy items 
As a result, in this setting, we override the notation and simply use $n$ to denote the total number of agents.
Formally, an exchange market is thus a tuple $(n, k, \bm{I}, \bm{F})$. 

In two-sided markets sellers are assumed to only value items in their initial bundle and are therefore not interested in buying from other sellers, i.e., $\forall j \in [m]$ and $\forall S \subseteq [k]$, $w_{j}(S) = w_{j}(S \cap I_{j})$. On the contrary, in exchange markets, no such restriction on the valuation functions exists.

Throughout the paper, we reserve the usage of the letter $i$ to denote a single buyer, the letter $j$ to denote a single seller, and the letter $\ell$ to denote a single item. Moreover, we use $v_i$ to denote buyer $i$'s valuation function and $w_j$ to denote seller $j$'s valuation function. 

\smallbreak 
\noindent \textbf{Mechanism Design Goals.}
In order to avoid overloading this section, the following content will only be described for two-sided markets (which are the main focus of the paper). Nonetheless, any of the subsequent concepts can naturally be extended to combinatorial exchange markets.

Given a two-sided market, our aim is to redistribute the items among the agents so as to maximize the \emph{social welfare} (the sum of the agents' valuations). An \emph{allocation} for a given two-sided market $(n, m, k, \bm{I}, \bm{G}, \bm{F})$ is a pair of vectors $(\bm{X},\bm{Y}) = ((X_1, \ldots, X_n),(Y_1,\ldots,Y_m))$ such that the union of $X_1, \ldots, X_n, Y_1, \ldots Y_m$ is $[k]$, and $X_1,\ldots X_n,Y_1, \ldots, Y_m$ are mutually non-intersecting. 

Redistribution of the items is done by running a \emph{mechanism} $\mathbb{M}$. A mechanism interacts with and receives input from the buyers, and outputs an \emph{outcome}, consisting of an allocation $(\bm{X},\bm{Y})$ and a payment vector $(\bm{\rho}^B, \bm{\rho}^S) \in \mathbb{R}^n \times \mathbb{R}^m$, where $\bm{\rho}^B$ refers to the buyers' vector of payments and $\bm{\rho}^S$ to the sellers' one. 
The outcome of a mechanism $\mathbb{M}$ is represented by a tuple $(\bm{X},\bm{Y},\bm{\rho}^B,\bm{\rho}^S)$.

Agents are assumed to maximize their \emph{utility}, which is defined as the valuation for the bundle of items that they get allocated, minus the payment charged by the mechanism. In particular, the utility $u_i^B(v_i,(\bm{X},\bm{Y},\bm{\rho}^B,\bm{\rho}^S))$ of a buyer $i \in [n]$ is $v_i(X_i) - \rho^B_i$, whereas for a seller $j \in [m]$ it is $u_j^S(w_j,(\bm{X},\bm{Y},\bm{\rho}^B,\bm{\rho}^S)) = w_j(Y_j) - \rho^S_j$.

Since agents are assumed to maximize their utility, they will strategically interact with the mechanism so as to achieve this. Our goal is to design mechanism such that there is a dominant strategy or Bayes-Nash equilibrium for the agents under which the mechanism returns an allocation with a high social welfare. For a allocation $(\bm{X}, \bm{Y})$, the social welfare $\SW(\bm{X}, \bm{Y})$ is defined as $\textsf{SW}(\bm{X}, \bm{Y}) = \sum_{i \in [n]} v_i(X_i) + \sum_{j \in [m]} w_j(Y_j)$.

The objective function we want to maximize is the above defined social welfare. We now describe three main constraints our mechanisms must adhere to. For each of these constraints we first introduce the strictest version and then a more relaxed one. Our mechanisms aim to satisfy the strictest versions, whenever possible.

\begin{itemize}
\item \emph{Dominant strategy incentive compatibility (DSIC):} It is a dominant strategy for every agent to report her true valuation sincerely, no matter what the others do, i.e., no agent can increase her utility by misreporting her true valuation.
		\item \emph{Bayesian incentive compatibility (BIC):} No agent can obtain a gain in her expected utility by declaring a valuation different from her true one, where the expectation is taken w.r.t. the others' valuations.\footnote{Technically, as can be inferred, the DSIC properties are reserved for \emph{direct revelation mechansims}, i.e., where the buyer solely interacts with the mechanism by reporting his valuation function. It is well-known that mechanisms admitting a dominant strategy can be transformed into DSIC direct revelation mechanisms, and those with a Bayes-Nash equilibrium can be transformed into BIC direct revelation mechanisms. This way, we extend the DSIC and BIC definitions to non-direct revelation mechanisms.}

		\item \emph{Ex-post individual rationality (ex-post IR):} It is not harmful for any agent to participate in the mechanism, i.e., there is guaranteed to be a strategy for an agent that yields the agent an increase in utility.
		\item \emph{Interim individual rationality (interim IR):} There is a strategy for each agent that yields her an expected increase in utility (where expectation is over the valuations of the other agents and the internal randomness of the mechanism).

		\item \emph{Strong Budget Balance (SBB):} The sum of all agents' payments output by the mechanism is equal to zero. Conceptually, this means that no money is burnt or ends up at an external party, and no external party needs to subsidise the mechanism.
		\item \emph{Weak Budget Balance (WBB):} The sum of all payments is at least zero. In two sided-markets this translates in having the buyers' payments being at least as large as the payments received by the sellers. 
\end{itemize}

For valuation profiles $(\bm{v}, \bm{w})$, $\textsf{OPT}(\bm{v}, \bm{w}) := \max_{\bm{X}, \bm{Y}}\{\textsf{SW}(\bm{X}, \bm{Y})\}$ denotes the \emph{optimal social welfare}, where the max goes over all feasible allocations $\bm{X}, \bm{Y}$. The \emph{expected optimal social welfare} is the value $\textsf{OPT} = \mathbb{E}_{\bm{v}, \bm{w}}[\textsf{OPT}(\bm{v}, \bm{w})]$. We say that a mechanism $\mathbb{M}$ \emph{$\alpha$-approximates the optimal social welfare} for some $\alpha > 1$ iff $\textsf{OPT} \leq \alpha \mathbb{E}_{\bm{v}, \bm{w}}[\mathsf{SW}(\mathbb{M}(\bm{v}, \bm{w}))]$. Our goal is to find mechanisms that $\alpha$-approximates the optimal social welfare for a low $\alpha$, are DSIC or BIC, SBB, and ex-post IR or interim IR.

\noindent \textbf{Classes of Valuation Functions.}
We will consider probability distributions over the following classes of valuation functions. Let $v : 2^{[k]} \rightarrow \mathbb{R}_{\geq 0}$ be a valuation function. Then,
\begin{itemize}
	\item $v$ is \emph{additive} iff there exists numbers $\alpha_1, \ldots \alpha_k \in \mathbb{R}_{\geq 0}$ such that $v(S) = \sum_{j \in S} \alpha_j$ for all $S \subseteq [k]$.
	\item $v$ is \emph{fractionally subadditive (or XOS)} iff there exists a collection of additive functions $a_1, \ldots, a_{d}$ such that for every bundle $S \subseteq [k]$ it holds that $v(S) = \max_{i \in [d]} a_i(S)$.
\end{itemize}
It is easy to see that every additive function is a XOS function. Further, it is well-known that the class of submodular functions are contained in the class of XOS functions. Due to space constraints, all proofs in the remainder of this paper have been deferred to the appendix.

\section{An Initial Mechanism and Direct Trade Strong Budget Balance}\label{sec:dsbb}
In \cite{bd14}, Blumrosen and Dobzinski present a mechanism for exchange markets with subadditive valuation functions. They prove the following for this mechanism, which we name $\mathbb{M}_{\text{bd}}$.
\begin{theorem}[Blumrosen and Dobzinski \cite{bd14}]
Mechanism $\mathbb{M}_{\text{bd}}$ is a DSIC, WBB, ex-post IR randomized direct revelation mechanism that $4H(s)$-approximates the optimal social welfare for combinatorial exchange markets $(n,k,\bm{I},\bm{F})$ with subadditive valuation functions, where $s = \min\{n, |I_i| : i \in [n]\}$ is the minimum of the number of agents and the number of items in an agent's initial endowment, and $H(\cdot)$ denotes the harmonic numbers.
\end{theorem}
In particular, this mechanism gives us a constant approximation factor if the number of starting items of the agents is bounded by a constant.

Now consider a mechanism $\mathcal{M}_{\text{sbb}}$ that selects an agent $i \in [n]$ uniformly at random, runs $\mathcal{M}_{\text{bd}}$ on the remaining agents, and allocates the surplus money of $\mathbb{M}_{\text{bd}}$ to agent $i$. We then are able to proof the following.
\begin{theorem}
Mechanism $\mathbb{M}_{\text{sbb}}$ is DSIC, ex-post IR, SBB, and achieves an $8nH(s)/(n-1)$-approximation to the optimal social welfare for exchange markets with subadditive valuations and at least $3$ agents.\footnote{For $2$ agents, it is straightforward to come up with alternative mechanisms that have the desired properties.}
\end{theorem}
The proof and a more precise description of $\mathbb{M}_{\text{sbb}}$ are provided in Appendix \ref{sec:comb_exchange}.
This yields an ex-post IR, SBB, DSIC mechanism that $O(1)$-approximates the social welfare if the number of items initially posessed by an agent is bounded by a constant.
The principle that we used to construct Mechanism $\mathbb{M}_{\text{sbb}}$ can more generally be used to turn any WBB mechanism into a SBB one, while preserving the DSIC and ex-post IR properties.
It also reveals a problematic aspect of the notion of SBB: it allows for agents to receive money, while they are not involved in any trade. This motivates a strengthened notion of strong budget balance, which we call \emph{direct trade strong budget balance}.
\begin{definition}
A mechanism for an exchange market satisfies \emph{direct trade strong budget balance (DSBB)} iff the outcome it generates can be achieved by a set of bilateral trades, where each trade consists of a reallocation of an item from an agent $i$ to an agent $j$, and a monetary transfer from agent $j$ to agent $i$. Moreover, each item may only be traded once.
\end{definition}
It can be seen that Mechanism $\mathbb{M}_{\text{sbb}}$ does not satisfy DSBB. In the remainder of the paper we will proceed to design mechanisms for two-sided markets that do satisfy DSBB.\footnote{We note that the double auctions given in \cite{cdlt16} also satisfy the DSBB property.} Moreover, two of our results provide an $O(1)$-approximation even in settings in which the $\mathbb{M}_{sbb}$ provides only a log-factor.

\section{A Mechanism for Unit-Supply Two-Sided Markets with XOS Buyers}\label{sec:unit-supply}
In this section we present a DSIC, ex-post IR, and DSBB mechanism for unit supply two-sided markets, when buyers have XOS valuation functions. This mechanism achieves constant approximation to the optimal social welfare. This result extends \citep{fgl15} to two-sided auctions and \citep{cdlt16} to a setting in which buyers have combinatorial preferences over the items.
In \emph{unit-supply two-sided markets}, for simplicity, we use $[k]$ to denote to both the set of items and the set of sellers, where item $j$ is owned by seller $j$ (so $I_j = \{j\}$ for all $j \in [k]$). For each seller $j \in [k]$, we then treat $F_j$ as a distribution over $\mathbb{R}_{\geq 0}$ instead of a distribution over functions.

We assume throughout this section that $(n,k,k,\bm{I},\bm{G},\bm{F})$ is a given unit-supply two-sided market, on which we run the mechanism to be defined. For an allocation $(\bm{X},\bm{Y}) \in \mathcal{A}$, we shall use the superscripts $\SW^B, \SW^S$ to respectively denote the buyers' and the sellers' contribution to the social welfare, i.e.,
\begin{equation*}
\SW^B(\bm{X}, \bm{Y}) := \sum_{i=1}^n v_i(X_i), \text{ and } \SW^S(\bm{X}, \bm{Y}) := \sum_{j=1}^k w_j \indicator{j \in Y_j}.
\end{equation*}




Our mechanism will work by specifying prices on the items in the market. For a bundle $\Lambda$ and an item price vector $\bm{p} = (p_1, \ldots, p_k)$, we define the \emph{demand correspondence} of buyer $i \in [n]$ with valuation function $v_i$ as 
\begin{equation*}
\mathcal{D}(v_i, \bm{p}, \Lambda) := \left\{S \subseteq \Lambda : v_i(S) - \sum_{j \in S} p_j \geq v_i(T) - \sum_{j \in T} p_j \text{ for all } T \subseteq \Lambda \right\},
\end{equation*}
i.e., the set of bundles of items in $\Lambda$ that maximize $i$'s utility under the given item prices, when she has valuation function $v_i$.

For a buyer $i$ with valuation function $v_i$, we define the \emph{additive representative function} for bundle $T \subseteq [k]$ as the additive function $a(v_i,T,\cdot) : 2^{[k]} \rightarrow \mathbb{R}_{\geq 0}$ such that $v_i(T) = a(v_i,T,T)$, and $v_i(S) \geq a(v_i,T,S)$ for all $S \subseteq [k]$ .
The additive representative function of a bundle is guaranteed to exist for each buyer $i$, for each valuation function in the support of $F_i$, by the definition of XOS functions.


\smallbreak
\noindent \textbf{Mechanism.}
Let $\mathbb{A}$ be an algorithm that, given a valuation profile of the buyers $\bm{v}$, allocates the items $[k]$ to them, and does not take into account the sellers and their valuations.
Our mechanism will use such an algorithm as a black-box, and it can be thought of as outputting either an allocation that is optimal for the buyers (in case one does not care about the runtime of the mechanism) or an approximately optimal one (in case one insists on the mechanism running in polynomial time).

Let $\bm{X}^\all(\bm{v}) = (X_1^\all(\bm{v}), \ldots, X_n^\all(\bm{v}))$ be the output allocation of $\mathbb{A}(\bm{v})$. 
Let $\SW(\bm{X}^\all(\bm{v}))$ be the total social welfare of the allocation $\bm{X}^\all(\bm{v})$.

We define for each item $j \in [k]$ its \emph{contribution $\SW^B_j(\bm{v})$ to the social welfare $\SW(\bm{X}^{\all}(\bm{v}))$} as follows: If there exists a player $i$ that receives item $j$ in allocation $X_i^\all(\bm{v})$, then $\SW^B_j(\bm{v}) = a(v_i, X_i^\all(\bm{v}), \{j\})$. Otherwise, if $j$ is not allocated to any buyer in $X_i^\all(\bm{v})$, then $\SW^B_j(\bm{v}) = 0$.

This notion allows us to make a distinction between \emph{high welfare items} and \emph{low welfare items}. An item $j \in [k]$ is said to have \emph{high welfare with respect to $\SW(X_i^\all(\bm{v}))$} iff $\mathbb{E}_{\bm{v}}[\SW^B_j(\bm{v})] \geq 4\expected{w_j}$, i.e., the expected social welfare contribution of $j$ if we would allocate $j$ according to $\bm{X}^\all(\bm{v})$ is at least four times as high as the social welfare that results from leaving item $j$ at its seller.

Formally, let $L$ be the set of high welfare items,
i.e., $L := \left\{\ell \in [k] : \expected{\SW^{B}_\ell(\bm{v}} \geq 4 \expected{w_{j}}\right\}$,
and let $\bar{L}$ be the set of low welfare items, i.e. $\bar{L} := [k]\setminus L$. For each high welfare item $j \in L$, the mechanism makes use of the following associated \emph{item price} $p_j$: 
\begin{equation*}
p_j :=\frac{1}{2}  \expectedsub{\bm{v}}{\SW^B_j(\bm{v})}.
\end{equation*}
Observe that $p_j \geq 2\mathbb{E}[w_j]$ for all $j \in L$, by our definition of high welfare items.
The reason why $L$ is chosen in such a way is twofold: first, the items in $\bar{L}$ if kept by their sellers provide a welfare loss of at most a constant factor; second, every item in $L$ is guaranteed to be sold (if sold) at a high price.

Our (randomized) mechanism does the following simple procedure. First, it goes to every seller in $L$ (in any order) and asks each of them whether they would sell their item for a price of $p_j$. As mentioned above, by definition of the prices, every seller $j \in L$ accepts the price with probability at least $1/2$ because of Markov's inequality. To make sure that this probability is exactly $1/2$, the seller $j$ is only given the opportunity to sell her item at the price $p_j$ with probability $q_j$ such that (in expectation) the offer is accepted with probability exactly $1/2$.
Formally, the mechanism makes an offer to the seller $j$ with probability
\[
q_j := \frac{1}{2 F_j(p_j)}, \text{ where } F_j(p_j) = \prob{w_j \leq p_j}.
\]
Once we have gone through every seller and know which items are in the market, we move to the buyers and ask each (in any order) for their favorite bundle according to the items currently available for purchase.

We call the mechanism sketched above $\mathbb{M}_{\text{1-supply}}$, which we will now present more precisely: \\
\noindent\fbox{%
\begin{varwidth}{\dimexpr\linewidth-2\fboxsep-2\fboxrule\relax}
\small{
\begin{enumerate}
	\item Let $L := \{j \in [k] : \mathbb{E}_{\bm{v}}[\SW^B_j(\bm{v})] \geq 4\mathbb{E}[w_j]\}$.
	\item For all $j \in L$, set $p_j :=\frac{1}{2} \expectedsub{\bm{v}}{\SW^B_j(\bm{v})}$.
	\item Let $\Lambda_1 := \emptyset, X_i := \emptyset$ for all $i \in [n]$ and $Y_j := \{j\}$ for all $j \in [k]$.
	\item For all $j \in L$:
	\begin{enumerate}
		\item Set $q_j := 1 / (2 \mathsf{Pr}[w_j \leq p_j])$.
		\item With probability $q_j$, offer payment $p_j$ in exchange for her item. Otherwise, skip this seller.
		\item If $j$ accepts the offer, set $\Lambda_1 := \Lambda_1 \cup \{j\}$.
	\end{enumerate}	
	\item For all $i \in [n]$:\label{algline:buyer_cycle}
	\begin{enumerate}
		\item Buyer $i$ chooses a bundle $B_i \in \mathcal{D}(v_i, \bm{p}, \Lambda_i)$ that maximizes her utility.
		\item Allocate the accepted items to buyer $i$, i.e., $X_i :=B_i$ and $Y_j := \emptyset$ for all $j \in B_i$.
		\item Remove the selected items from the available items, i.e., $\Lambda_{i+1} := \Lambda_i \setminus B_i$. 
	\end{enumerate}
	\item Return the outcome consisting of allocation $(\bm{X} = (X_1, \ldots, X_n), \bm{Y} = (Y_1, \ldots, Y_k))$ and payments $\bm{\rho} = (\bm{\rho}^B, \bm{\rho}^S)$, where $\rho_i^B = \sum_{j \in X_i} p_j$ for $i \in [n]$ and $\rho_j^S = -p_j \indicator{Y_j = \varnothing}$ for $j \in [k]$.
\end{enumerate}
}
\end{varwidth}
}\\
Note that Algorithm $\mathbb{A}$ is used in the first steps of mechanism $\mathbb{M}_{\text{1-supply}}$, where $\mathbb{E}_{\bm{v}}[\SW^B_j(\bm{v})]$ is computed.
Let $\alpha$ be the factor by which $\mathbb{A}$ is guaranteed to approximate the social welfare of the buyers. 

\begin{theorem}
\label{thm:18-approx_xos}
$\mathbb{M}_{\text{1-supply}}$ is ex-post IR, DSIC, DSBB, and $(2 + 4\alpha)$-approximates the optimal social welfare. 
\end{theorem}
In particular, taking for $\mathbb{A}$ an optimal algorithm (i.e., $\alpha = 1$), we obtain that there exists a mechanism that is ex-post IR, DSIC, DSBB, and $6$-approximates the optimal social welfare.
Alternatively, one may take for $\mathbb{A}$ an approximation algorithm in order to obtain polynomial time implementable mechanisms, as explained in further detail in \cite{fgl15}.

\section{A Mechanism for Two-Sided Markets with XOS Buyers and Additive Sellers}
We now consider the setting in which sellers may own multiple distinct items and have an additive valuation function over them. We design a DSBB mechanism that is DSIC and ex-post IR on the sellers' side, and BIC and interim IR on the buyers' side. 
At the end of the section also proof that in the case of both buyers and sellers possessing additive valuation functions, the mechanism we present is DSIC and ex-post IR on both sides of the market.

We assume throughout this section that $(n,m,k,\bm{I},\bm{G},\bm{F})$ is a given two-sided market with XOS buyers and Additive Sellers, on which we run the mechanism to be defined. 
Like in the previous section, the buyers are still assumed to have XOS valuation functions over the items. Since now the number of items and sellers is different in general, we use $m$ to denote the number of sellers and $k$ for the number of items. The valuation $w_j$ of a seller $j$ is now an additive function.
We reuse the following notation from Section \ref{sec:unit-supply}: The allocation $(X_1^\all(\bm{v}), \ldots, X_n^\all(\bm{v}))$ returned by an allocation algorithm $\mathbb{A}$ on input $\bm{v}$ returns an allocation of $[k]$ to $[n]$. We let $\alpha \geq 1$ again denote the approximation factor by which $\mathbb{A}$ approximates the social welfare. For XOS valuation $v_i$ and bundle $T \subseteq [k]$ we use $a(v_i,T, \cdot)$ to denote the corresponding additive function of $v_i$ for $T$. Also we use the buyers' social welfare contribution $\SW^B_\ell(\bm{v})$ for item $\ell \in [k]$ and buyers' valuation profile $\bm{v}$, as defined in Section \ref{sec:unit-supply}.
We define furthermore the \emph{sellers' social welfare contribution} $\SW^S_\ell(\bm{w})$ for item $\ell \in I_j$ and sellers' valuation profile $\bm{w}$ as $\SW^S_\ell(\bm{w}) := w_j(\{\ell\})$.
Due to the fact that for $j \in [m]$, $w_j$ is an additive function, there is no need for defining the notion of a \emph{corresponding additive function} for a seller.

\smallbreak
\noindent \textbf{Mechanism.}
We aim to design a BIC, interim IR, and SBB mechanism that approximates the optimal social welfare within a constant. We propose the following mechanism, which we refer to as $\mathbb{M}_\text{add}$.
We let $L_j := \{\ell \in I_j : \mathbb{E}[\SW^B_\ell(\bm{v})] \geq 4 \mathbb{E}[\SW^S_\ell(\bm{w})]\}$ and $\bar{L}_j := I_j \setminus L_j$ for all $j \in [m]$, and we let $L := \bigcup_{j = 1}^m L_j$ and $\bar{L} := [k] \setminus L$ denote the sets of high welfare items and low welfare items, respectively. Our mechanism will only allow trading items in $L$, and define for $\ell \in L$ the item price 
\begin{equation*}
p_\ell := \frac{1}{2} \expected{\SW^B_\ell(\bm{v})},
\end{equation*}
similar to what we did for $\mathbb{M}_{\text{1-supply}}$.

An essential difference between $\mathbb{M}_{\text{add}}$ and $\mathbb{M}_{\text{1-supply}}$ is that the order in which buyers and sellers are processed is reversed.
Mechanism $\mathbb{M}_{\text{add}}$ roughly works as follows. It first asks every buyer which set of items it would like to receive from those items in $L$ that have not been requested yet.
Then $\mathbb{M}_{\text{add}}$ offers every seller $j \in [m]$ a payment in exchange for the subset of all items in $I_j$ that have been requested. This offer is made with a $q_j$ ensures that the requested items of seller $j$ are transfered to the buyers with probability $1/2$. The items of the sellers accepting the offer are transferred to the buyers for the corresponding item prices. Buyers act strategically, and request a set of items that maximizes their expected utility, knowing that the item sets requested from each seller will be assigned to them with probability $1/2$\footnote{The buyer may need to compute complicated expected values in order to establish which bundle is maximizing his expected utility. But this depends on the query oracle model we assume to have. However, most importantly, the focus of the paper is not on polynomial time implementability but on the existence of mechanisms that constantly approximate the optimal social-welfare.}.

\noindent\fbox{%
\begin{varwidth}{\dimexpr\linewidth-2\fboxsep-2\fboxrule\relax}
\small{
\begin{enumerate}
\item For $\ell \in [k]$, compute $\expected{\SW^B_\ell(\bm{v})}$ and $\expected{\SW^S_\ell(\bm{w})}$.
\item For all $j \in [m]$, compute $L_j$ and $q_j$.
\item Compute $L$, $\bar{L}$. 
\item Let $\Lambda_1 := L$, $X_i := \varnothing$ for all $i \in [n]$, and $Y_j := I_j$ for all $j \in [m]$.
\item \label{buyerloop} For each buyer $i \in [n]$:
	\begin{enumerate}
	\item Ask buyer $i$ to select an expected-maximising bundle $B_i \subseteq \Lambda_i$ given the prices $\{p_\ell : \ell \in \Lambda_i\}$ from the set of available items (where the expectation is taken with respect to the randomness of both the mechanism and the valuations).
	\item \label{lambda} Update the set of available items $\Lambda_{i+1} := \Lambda_i \setminus B_i$.
	\end{enumerate}
\item Let $B := \bigcup_{i=1}^n B_i$ be the set of all items demanded by the buyers.
\item For each seller $j \in [m]$:
	\begin{enumerate}
	\item Let $S_j := B \cap L_j$ be the set of items owned by seller $j$ that are demanded.
	\item Let $p(S_j) := \sum_{\ell \in S_j} p_\ell$ and let $q_j = 1/(2 \prob{w_j(S_j) \leq p(S_j)})$.
	\item \label{offer} With probability $q_j$, offer payment $p(S_j)$ in exchange for the bundle $S_j$. Otherwise, skip this seller.
	\item If the seller accepts the offer, allocate each items in $S_j$ to the buyer that requested it (i.e., remove $S_j$ from $Y_j$ and add $S_j \cap B_i$ to $X_i$ for all $i \in [n]$) 
	\end{enumerate}
\item Return the outcome consisting of allocation $(\bm{X} = (X_1, \ldots, X_n), \bm{Y} = (Y_1, \ldots, Y_k))$ and payments $\bm{\rho} = (\bm{\rho}^B, \bm{\rho}^S)$, where $\rho_i^B = \sum_{\ell \in X_i} p_\ell$ for $i \in [n]$ and $\rho_j^S = \sum_{\ell \in I_j \setminus Y_j} -p_\ell$ for $j \in [m]$.
\end{enumerate}
}
\end{varwidth}
}\\

\begin{example}
There is $1$ buyer and $2$ unit-supply sellers. Each seller posses $1$ item.
The buyer has two XOS valuation functions $v_{1}$ and $v_{2}$, each one is chosen with probability $0.5$.
$v_{1}$ is composed by $3$ additive functions $a_{1}$, $a_{2}$, and $a_{3}$, i.e., $v_{1}(S)=\max \{a_{1}(S), a_{2}(S), a_{3}(S)\}$.
$v_{2}$ is composed by the only additive function $a_{4}$. The additive functions are represented in the following table, recall that an additive function is represented by $a(S) = \sum_{j \in S} \alpha_j$ for all $S \subseteq [k]$.
Each seller $j$ has a valuation function $w_j = 0$, so we always want to reallocate items from sellers to the buyer.

\begin{center}
    \begin{tabular}{ | c | c | c | }
    \hline
    Function & item $1$ ($\alpha_{1}$)& item $2$ ($\alpha_{2}$)\\ \hline
    $a_{1}$ & $0$ & $4$ \\ \hline
    $a_{2}$ & $8$ & $0$ \\ \hline
    $a_{3}$ & $7$ & $2$ \\ \hline
    $a_{4}$ & $1$ & $6$ \\
    \hline
    \end{tabular}
\end{center}

Now, we have to compute the prices that the mechanism has to post. Thus, we need the expected contribution to the optimal social welfare of every item.
First, notice that the optimum allocates the items $1$ and $2$ to the buyer when her valuation is $v_{1}$. In this case the contribution to the optimal social welfare of items $1$ is $7$, and the contribution of item $2$ is $2$.
Similarly, if the buyer has valuation $v_{2}$, the optimum still allocates items $1$ and $2$ to her, but in this case the contribution to the optimal social welfare of item $1$ is $1$, and the contribution of item $2$ is $6$.
Thus, the expected contribution of every item to the optimal social welfare is $4$, i.e., $\mathbb{E}[\SW^B_j(\bm{v})] = 4$ for all $j =1,2$.
Since the price $p_j$ of each item is defined to be half of the expected contribution to the optimal social welfare,  $p_{j} = 2$ for all the items.

Now the buyer has to compute the bundle that maximises her expected-utility given the posted prices, and the fact that each seller is available with probability $0.5$.
First, consider the case when the buyer has valuation $v_{1}$. In this case the expected utility for the different bundles are:
\[
u(\{1\})=\frac{1}{2} \cdot (8-4) + \frac{1}{2} \cdot 0 = 2
\]
\[
u(\{2\})=\frac{1}{2} \cdot (4-4) + \frac{1}{2} \cdot 0 = 0
\]
\[
u(\{1,2\}) = \frac{1}{4} \cdot (8-4) + \frac{1}{4} \cdot (4-4) +\frac{1}{4} \cdot (9 - 8) + \frac{1}{4} \cdot 0 = \frac{5}{4}
\]
The utility-maximising bundle that will be requested by the buyer in case of $v_{1}$ is $\{1\}$.
Instead, if the valuation of the buyer is $v_{2}$, then the requested bundle will be $\{1,2\}$.
\end{example}

This example shows that the buyers are committed to difficult computations in order to understand which is the utility-maximising bundle. But, again, our goal is to understand the existence of mechanisms that achieve a constant approximation to the optimum social welfare, even if they are not computational efficient.

\begin{theorem}\label{thm:bicmech}
The mechanism $\mathbb{M}_\mathrm{add}$ is ex-interim IR, BIC, DSBB, and $(2 + 4\alpha)$-approximates the optimal social welfare.
\end{theorem}
By taking for $\mathbb{A}$ an optimal algorithm (i.e., $\alpha = 1$), we obtain the existence of a mechanism that is ex-post IR, DSIC, DSBB, and $6$-approximates the optimal social welfare.
However, if both a polynomial time approximation algorithm $\mathbb{A}$ and an approximation for the query oracle exist, then we get our mechanism to run in polynomial time.
\begin{corollary}\label{cor:additive}
For the special case that for all $i \in [n]$, distribution $G_i$ is over additive valuation functions, $\mathbb{M}_\mathrm{add}$ is ex-post IR, DSIC, DSBB and $(2 + 4\alpha)$-approximates the optimal social welfare.
\end{corollary}

\section{Discussion}
An open problem is to extend or refine our mechanisms so that they satisfy the DSIC and ex-post IR properties for the case of XOS buyers and additive sellers. The first naive approach for doing so might be trying to consider every additive seller as a set of distinct unit-supply sellers and then run $\mathbb{M}_\text{1-supply}$. However, this is not guaranteed to work due to the fact that the items placed in the market by the seller influence the buyers' decisions on their bundle choices. 
Something we might additionally do is to ask every seller for her favourite bundle to place in the market, 
However, this may cause a seller to regret having chosen that particular bundle after seeing the realizations of the buyers' valuations. On the other hand, it also seems highly challenging to establish any sort of impossibility result for any reasonably defined class of posted price mechanisms for two-sided markets.

Another natural challenge is to extend the above mechanism to the setting in which both buyers and sellers possess an XOS valuation function over bundles of items. 
We suspect however that it is impossible to devise a suitable mechanism that uses item-pricing (i.e., mechanisms that fix a price vector offline and then request buyers or sellers one-by-one to specify their favorite item bundle).
The main reason for why this seems to be the case is that if a buyer (or a seller) is asked to select a bundle of items $B$ she desires the most, 
then it is impossible to guarantee that she receives the complete bundle. Instead, she may receive a subset of it, and in turn she may regret having chosen that bundle $B$ and not another bundle $B'$ after having observed the realizations of the sellers' valuations. 


\bibliographystyle{splncsnat}
\bibliography{literature}

\appendix

\section{An Initial Mechanism and Direct Trade Strong Budget Balance (Full Details)}\label{sec:comb_exchange}
In \cite{bd14}, Blumrosen and Dobzinski present a mechanism for exchange markets with subadditive valuation functions. They prove the following for this mechanism, which we name $\mathbb{M}_{\text{bd}}$.
\begin{theorem}[Blumrosen and Dobzinski \cite{bd14}]\label{thm:db-sw}
Mechanism $\mathbb{M}_{\text{bd}}$ is a DSIC, WBB, ex-post IR randomized direct revelation mechanism that $4H(s)$-approximates the optimal social welfare for combinatorial exchange markets $(n,k,\bm{I},\bm{F})$ with subadditive valuation functions, where $s = \min\{n, |I_i| : i \in [n]\}$ is the minimum of the number of agents and the number of items in an agents initial endowment, and $H(\cdot)$ denotes the harmonic numbers.
\end{theorem}
In particular, this mechanism gives us a constant approximation factor if the number of starting items of the agents is bounded by a constant.

We show now how we can use this mechanism as a black box in order to obtain an SBB mechanism with only a slightly worse approximation ratio.
Define mechanism $\mathcal{M}_{\text{sbb}}$ as follows. When given as input a combinatorial exchange market $C = (n,k,\bm{I},\bm{F})$,
\begin{enumerate}
\item Select an agent in $i \in [n]$ uniformly at random.
\item Run Mechanism $\mathbb{M}_{\text{bd}}$ on the combinatorial exchange market 
\begin{equation*}
C_{-i} = ([n]\setminus\{i\}, \bm{I}_{-i} = (I_1, \ldots, I_{i-1}, I_{i+1}, \ldots, I_n), \bm{F}_{-i} = (F_1, \ldots, F_{i-1}, F_{i+1}, \ldots, F_n)). 
\end{equation*}
Let $(\bm{X}_{-i},\bm{\rho}_{-i})$ be the outcome that Mechanism $\mathbb{M}_{\text{bd}}$ outputs.
\item Set $X_i = I_i$ and set $p_i = -\sum_{j \in [n] \setminus\{i\}} p_j$. Output the allocation $(X_i, \bm{X}_{-i})$ and output payment vector $(p_i, \bm{\rho}_{-i})$.
\end{enumerate}
So Mechanism $\mathbb{M}_{\text{sbb}}$ essentially runs Mechanism $\mathbb{M}_{\text{bd}}$ where one random agent is removed from the market. This agent receives the leftover money that Mechanism $\mathbb{M}_{\text{bd}}$ generates, and does not receive or lose any items.
The following is a direct corollary of the DSIC, WBB, and ex-post IR properties of mechanism $\mathbb{M}_{\text{bd}}$.
\begin{theorem}
Mechanism $\mathbb{M}_{\text{sbb}}$ is a DSIC, SBB, and ex-post IR mechanism for exchange markets with subadditive valuation functions.
\end{theorem}
Secondly, the following theorem shows that the mechanism loses only a factor $2n/(n-1) \leq 3$ in the approximation ratio for $n \geq 3$. (For $n = 2$ it is straightforward to come up with alternative mechanisms that achieve a good approximation ratio.)
\begin{theorem}
Mechanism $\mathbb{M}_{\text{sbb}}$ achieves an $8nH(s)/(n-1)$-approximation to the optimal social welfare for exchange markets with subadditive valuations and at least $3$ agents.
\end{theorem}
\begin{proof}
Fix a valuation vector $\bm{v}$ of the agents, let $\bm{X}_{\bm{v}}^{**} \subseteq \mathcal{A}$ be the social welfare maximising allocation when the agents have valuations $\bm{v}$. For an agent $i \in [n]$, denote by $\bm{X}_{\bm{v},-i}^{**}$ the allocation for $C_{-i}$ where $(X_{\bm{v},-i}^{**})_j = (X_{\bm{v}}^{**})_j \setminus I_i$ for $j \in [n] \setminus \{i\}$, i.e., the allocation obtained from $\bm{X}_{\bm{v}}^{**}$ when $i$ is removed, and all items of $i$ are removed. Moreover let $\bm{X}_{\bm{v},-i}^{*}$ be the optimal allocation of the combinatorial exchange market $C_{-i}$ when the valuation function vector of the players $[n]\setminus\{i\}$ is fixed to $\bm{v}_{-i}$. Mechanism $\mathbb{M}_{\text{sbb}}$ selects $i$ uniformly at random, so by Theorem \ref{thm:db-sw}, the expected social welfare of Mechanism $\mathbb{M}_{\text{sbb}}$ is at least

\begin{eqnarray*}
\frac{1}{4H(s)}\expectedsub{i}{\sum_{j \in [n] \setminus \{i\}} v_j(\bm{X}_{\bm{v},-i}^{*})} & \geq &
\frac{1}{4H(s)}\expectedsub{i}{\sum_{j \in [n] \setminus \{i\}} v_j(\bm{X}_{\bm{v},-i}^{**})}\\
& = & \frac{1}{4nH(s)} \sum_{i \in [n]} \sum_{j \in [n] \setminus \{i\}} v_j((X_{\bm{v}}^{**})_j \setminus I_i) \\
& = & \frac{1}{4nH(s)} \sum_{i \in [n]} \sum_{j \in [n] \setminus \{i\}} v_i((X_{\bm{v}}^{**})_i \setminus I_j) \\
& = & \frac{1}{4nH(s)} \sum_{i \in [n]} \sum_{\substack{\{j,j'\} : j, j' \in [n] \setminus \{i\} \\ \wedge j \not= j'}} \frac{1}{n-2}(v_i((X_{\bm{v}}^{**})_i \setminus I_j) + v_i((X_{\bm{v}}^{**})_i \setminus I_{j'}))\\
& \geq & \frac{1}{4nH(s)} \sum_{i \in [n]} \sum_{\{j,j'\} : j, j' \in [n] \setminus \{i\} \wedge j \not= j'} \frac{1}{n-2}v_i(\bm{X}_{\bm{v}}^{**}) \\
& = & \frac{1}{4nH(s)} \sum_{i \in [n]} \frac{n-1}{2}v_i(\bm{X}_{\bm{v}}^{**}) \\
& = & \frac{n-1}{8nH(s)} \sum_{i\in [n]} v_i(\bm{X}_{\bm{v}}^{**}), 
\end{eqnarray*}
\noindent where the second inequality follows from subadditivity. This proves the claim, since the above holds for every valuation vector $\bm{v}$. \qed
\end{proof}
This yields an ex-post IR, SBB, DSIC mechanism that $O(1)$-approximates the social welfare if the number of items initially posessed by an agent is bounded by a constant.


The principle that we used to construct Mechanism $\mathbb{M}_{\text{sbb}}$ can more generally be used to turn any WBB mechanism into a SBB one, while preserving the DSIC and ex-post IR properties.
This principle also reveals a problematic aspect of the notion of SBB: it allows for agents to receive money, while they are not involved in any trade. This motivates a strengthened notion of strong budget balance, which we call \emph{direct trade strong budget balance}.
\begin{definition}
A mechanism for an exchange market satisfies \emph{direct trade strong budget balance (DSBB)} iff the outcome it generates can be achieved by a set of bilateral trades, where each trade consists of a reallocation of an item from an agent $i$ to an agent $j$, and a monetary transfer from agent $j$ to agent $i$. Moreover, each item may only be traded once.
\end{definition}
It can be seen that Mechanism $\mathbb{M}_{\text{sbb}}$ does not satisfy DSBB. In the remainder of the paper we will proceed to design mechanisms for two-sided markets that do satisfy DSBB.\footnote{We note that the double auctions given in \cite{cdlt16} also satisfy the DSBB property.}

\section{Proof of Theorem \ref{thm:18-approx_xos}}

We split the proof into two lemmas that bound the sellers' and the buyers' relative contributions to the social welfare. We use the notation \OPT \ as definied in Section \ref{sec:prelims}, and we use \ALG \ to denote the expected social welfare of the mechanism, i.e., $\mathbb{E}_{\bm{v},\bm{w}}[\SW(\mathbb{M}_\text{1-supply}(\bm{v},\bm{w}))]$. Moreover, the superscripts $S, B$ respectively denote the sellers' and buyers' contributions to the social welfare, e.g., $\OPT = \OPT^S + \OPT^B$ and $\ALG = \ALG^S + \ALG^B$. 
Recall that we use $\bm{X}^{\all}(\bm{v})$ to denote the allocation resulting from Algorithm $\mathbb{A}$ on buyers' valuation vector $\bm{v}$.


The following lemma is a simple consequence of the fact that $\mathbb{M}_\text{1-supply}$ lets every seller in $L$ gets an offer and accepts it with probability exactly$1/2$.
\begin{lemma}
\label{lemma:ALG^S}

If every seller $j \in L$ puts her item into the market with probability exactly $1/2$, then

\begin{equation*}
2 \ALG^S \geq \sum_{j =1}^k \expected{w_j} \geq \OPT^S.
\end{equation*}
\end{lemma}
\begin{proof}

The second inequality is trivial, so we focus on the first inequality.
First, observe that

\[
\prob{w_j > p_j} \leq \prob{w_j > 2 \expected{w_j}} < \frac{1}{2},
\]

where the first inequality is because $j \in L$, and the second inequality is by Markov's inequality.
Thus with probability at least $1/2$ a seller $j$ is happy to sell his item at price $p_{j}$.
But every seller receives an offer from the mechanism with probability $q_j := 1 / (2 \mathsf{Pr}[w_j \leq p_j])$, so every seller gets an offer and accepts it with
probability exactly $1/2$.
This implies that every $j \in L$ contributes in expectation at least a $\expected{w_j}/2$ to the social welfare, since with non negative probability the buyers won't buy
the item $j$.
Moreover, a seller in $\bar{L}$ who never trades, so their full expected valuations are contributed to the expected social welfare. 
\qed
\end{proof}

Next, we prove a more difficult bound that relates $\ALG^B$ and $\ALG^S$ to $\OPT^B$.
\begin{lemma}\label{lemma:ALG^B} 
The buyers' contributions to the optimal social welfare is bounded by

\[
4\alpha \ALG^B + 4\alpha \ALG^S \geq \OPT^B.
\]
\end{lemma}

%
%

\noindent Before proving Lemma \ref{lemma:ALG^B}, we point out that \refTheorem{thm:18-approx_xos} follows straightforwardly from it.
\begin{proof}[of \refTheorem{thm:18-approx_xos}]
The bound on the approximation ratio follows from the sum of the inequalities of \refLemma{lemma:ALG^S} and \refLemma{lemma:ALG^B}. Moreover, it is a dominant strategy for a seller to accept if and only if the payment offered to her exceeds her valuation, and it is a dominant strategy for a buyer to choose a utility-maximising bundle for the items and item prices offered to her. Thus, when viewed as a direct revelation mechanism, $\mathbb{M}_{\text{1-supply}}$ is DSIC. It is clear that participating in the mechanism can never lead to a decrease in utility for both buyers and sellers, and therefore the mechanism is also ex-post IR. Lastly, it is straightforward to see that the mechanism is DSBB, as the definition of $\mathbb{M}_{\text{1-supply}}$ which we gave in terms of sequential posted pricing naturally yields us the required set of bilateral trades.
\qed
\end{proof}
So it remains to prove Lemma \ref{lemma:ALG^B}. In order to do this, we first prove two propositions: one of them bounds the expected sum of the buyers' utilities, and one of them bounds the expected sum of the buyers' payments. 
In both propositions we only consider items in $L$ 
Given a buyers' valuation profile $\bm{v}$, let $\bm{v}_{<i} = (v_1, \ldots, v_{i-1})$.
Further, let $Z$ be a random variable that denotes the sellers that receive and accept an offer from the mechanism, i.e., the set $\Lambda_{1}$ at step \ref{algline:buyer_cycle} of $\mathbb{M}_{1-supply}$. For $i \in [n]$ let $\Lambda_i(\bm{v}_{<i}, Z)$ be the set $\Lambda_i$ as given in the definition of $\mathbb{M}_{\text{1-supply}}$ when the valuation profile of the buyers is $\bm{v}$ and $Z$ are the sellers in the market. Note that this implies that $X_i \subseteq \Lambda_i(\bm{v}_{<i}, Z) \subseteq Z$. Consequently, $\Lambda_{n+1}(\bm{v}, Z)$ is the subset of items for which the corresponding sellers have accepted the offer made to them by the mechanism, but remain allocated to the corresponding seller after execution.

\begin{proposition}
\label{prop:utility_B}
The total expected utility of the buyers for the allocation returned by $\mathbb{M}_{\text{1-supply}}$ is bounded from below by

\begin{equation*}
\mathbb{E}\left[\sum_{i \in [n]} u_i(\mathbb{M}_{\text{1-supply}}(\bm{v},\bm{w}))\right] \geq \frac{1}{2} \sum_{j \in L} \probsub{\bm{v}, Z}{j \in \Lambda_{n+1}(\bm{v}, Z)\ |\ j \in Z} p_j.
\end{equation*}
(Note that the random variables in this expression are $\bm{v}, \bm{w}$, and the decisions of the mechanism to make offers to the sellers in $L$.)
\end{proposition}

\begin{proof}
First, note that for each $j \in L$ it holds that $\mathsf{Pr}[j \in Z] = 1/2$. Recall that we defined $p_j := \frac{1}{2} \mathbb{E}_{\bm{v}}[\SW^B_j(\bm{v})]$. Thus, observe that by definition of $p_j$, $\SW_j^B(\bm{v})$, and the law of total probability, it holds for all $j \in L$ that
\begin{equation}
\label{eq:pj_sw}
p_j = \expectedsub{\bm{v}}{\SW^B_j(\bm{v}) - p_j}  = \sum_{i=1}^n \expectedsub{\bm{v}}{(\SW^B_j(\bm{v}) - p_j) \indicator{j \in X_i^\all(\bm{v})}}.
\end{equation}

Fix $i \in [n]$, buyers' valuation profile $\bm{v}$, and set $Z \subseteq L$ of sellers who accepted the mechanism's offer, and now consider the set $\Lambda_i(\bm{v}_{<i}, Z) \subseteq L$ of available items that $i$ can choose from. 
Buyer $i$ selects a bundle that maximizes her utility, i.e., that is in $\mathcal{D}(v_i, \bm{p}, \Lambda_i(\bm{v}_{<i}, Z))$.

Now consider an additional randomly drawn profile of valuation functions $\tilde{\bm{v}}_{-i}$ for all buyers except $i$, that is independent of $\bm{v}$. 
Consider the allocation $X^\all_i(v_i, \tilde{\bm{v}}_{-i})$ be the allocation of buyer $i$ returned by $\mathbb{A}(v_i, \tilde{\bm{v}}_{-i})$. For $i \in [n]$, consider the corresponding additive representative function $a(v_i, X_i^\all(v_i, \tilde{\bm{v}}_{-i}), \cdot)$, such that $a(v_i, X_i^\all(v_i, \tilde{\bm{v}}_{-i}),\{j\}) = \SW^B_j(v_i, \tilde{\bm{v}}_{-i})$. Let 
\begin{equation*}
S_i(v_i, \bm{v}_{-i}, \tilde{\bm{v}}_{-i}, Z) := X_i^\all(v_i, \tilde{\bm{v}}_{-i}) \cap \Lambda_i(\bm{v}_{<i}, Z)
\end{equation*}
be the items in $X_i^\all(v_i, \tilde{\bm{v}}_{-i})$ that buyer $i$ may choose from under valuation profile $\bm{v}$. As $i$ chooses a bundle $B_i(\bm{v}, Z) \in \mathcal{D}(v_i, \bm{p}, \Lambda_i(\bm{v}_{<i}, Z))$ that maximizes her utility, and $S_i(v_i, \bm{v}_{-i}, \tilde{\bm{v}}_{-i}, Z)$ is in $\mathcal{D}_i(v_i, \bm{p}, \Lambda_i(\bm{v}_{<i}, Z))$, it follows that $i$'s utility for $B_i(\bm{v}, Z)$ is at least the utility she would get for choosing $S_i(v_i, \bm{v}_{-i}, \tilde{\bm{v}}_{-i}, Z)$. That is, for all $\bm{v}$ and $Z \subseteq L$
\begin{eqnarray*}
v_i(B_i(\bm{v}, Z)) - \sum_{j \in B_i(\bm{v},Z)} p_j & \geq & \mathbb{E}_{\tilde{\bm{v}}_{-i}}\left[v_i(S_i(v_i, \bm{v}_{-i}, \tilde{\bm{v}}_{-i}, Z)) - \sum_{j \in S_i(v_i, \bm{v}_{-i}, \tilde{\bm{v}}_{-i}, Z)} p_j\right] \\
& \geq &  \expectedsub{\tilde{\bm{v}}_{-i}}{\sum_{j \in S_i(v_i, \bm{v}_{-i}, \tilde{\bm{v}}_{-i}, Z)} (a(v_i, X_i^\all(v_i, \tilde{\bm{v}}_{-i}), \{j\}) - p_j)}\\
& = &  \expectedsub{\tilde{\bm{v}}_{-i}}{\sum_{j \in S_i(v_i, \bm{v}_{-i}, \tilde{\bm{v}}_{-i}, Z)} (\SW^B_j(v_i, \tilde{\bm{v}}_{-i}) - p_j)},
\end{eqnarray*}
The second-to-last inequality follows from the definition of the corresponding additive function $a(v_i, X_i^\all(v_i, \tilde{\bm{v}}_{-i}), \cdot)$. 

Now summing the above expression over all $i \in [n]$ and taking the expectation over $\bm{v}$ and $Z$, we get
\begin{eqnarray*}
\expectedsub{\bm{v}, Z}{\sum_{i=1}^n \left(v_i(B_i(\bm{v}, Z)) - \sum_{j \in B_i(\bm{v},Z)} p_j\right)} & \geq & \expectedsub{\bm{v}, \tilde{\bm{v}}_{-i}, Z}{\sum_{i=1}^n \sum_{j \in S_i(v_i, \bm{v}_{-i}, \tilde{\bm{v}}_{-i}, Z)} (\SW^B_j(v_i, \tilde{\bm{v}}_{-i}) - p_j)}\\
	& = & \expnobrack_{\bm{v}, \tilde{\bm{v}}_{-i}, Z}\Bigg[\sum_{i=1}^n \sum_{j \in L} (\SW^B_j(v_i, \tilde{\bm{v}}_{-i}) - p_j) \\
 	&& \ \cdot \indicator{j \in X_i^\all(v_i, \tilde{\bm{v}}_{-i})} \indicator{j \in \Lambda_i(\bm{v}_{<i}, Z)}\Bigg].
\end{eqnarray*}
Note that we exploited the independence of the events $(j \in X_i^\all(v_i, \tilde{\bm{v}}_{-i}))$ and $(j \in \Lambda_i(\bm{v}_{<i}, \bm{z}))$. 
Thus, switching the order of the sums and using linearity of expectation, we get that
\begin{align*}
& \expectedsub{\bm{v}, Z}{\sum_{i=1}^n \left(v_i(B_i(\bm{v}, Z)) - \sum_{j \in B_i(\bm{v},Z)} p_j\right)} \\
& \qquad \geq \sum_{j \in L} \sum_{i=1}^n \probsub{\bm{v}, Z}{j \in \Lambda_i(\bm{v}_{<i}, Z)}\expectedsub{v_i, \tilde{\bm{v}}_{-i}}{(\SW^B_j(v_i, \tilde{\bm{v}}_{-i}) - p_j) \indicator{j \in X_i^\all(v_i, \tilde{\bm{v}}_{-i})}} \\
& \qquad \geq \sum_{j \in L} \probsub{\bm{v}, Z}{j \in \Lambda_{n+1}(\bm{v}, Z)} \sum_{i=1}^n \expectedsub{\bm{v}}{(\SW^B_j(\bm{v}) - p_j) \indicator{j \in X_i^\all(\bm{v})}} \\
& \qquad = \sum_{j \in L} \probsub{\bm{v}, Z}{j \in \Lambda_{n+1}(\bm{v}, Z)} p_j \\
& \qquad = \sum_{j \in L} \probsub{\bm{v}, Z}{j \in \Lambda_{n+1}(\bm{v}, Z)\ |\ j \in Z}\prob{j \in Z} p_j \\
& \qquad = \frac{1}{2}\sum_{j \in L} \probsub{\bm{v}, Z}{j \in \Lambda_{n+1}(\bm{v}, Z)\ |\ j \in Z}p_j.
\end{align*}
For the last inequality, 
we used the fact that for any $i \in [n]$ it holds that $\probsub{\bm{v}}{j \in \Lambda_i(\bm{v}_{<i}, Z)} \geq \probsub{\bm{v}}{j \in \Lambda_{n+1}(\bm{v}, Z)}$. The first equality follows from (\ref{eq:pj_sw}). \qed
\end{proof}

\begin{proposition}
\label{prop:revenue_B}
The expected sum of the payments charged by $\mathbb{M}_\text{1-supply}$ to the buyers is equal to
\begin{equation*}
\expected{\sum_{i \in [n]} \rho_i^B} = \frac{1}{2} \sum_{j \in L} p_j \probsub{\bm{v}, Z}{j \notin \Lambda_{n+1}(\bm{v}, Z)\ |\ j \in Z}
\end{equation*}
\end{proposition}

\begin{proof}
The revenue extracted by the mechanism, meaning the sum of the payments charged to the buyers, is equal to
\begin{eqnarray*}
\sum_{j \in L} p_j \probsub{\bm{v}, Z}{j \notin \Lambda_{n+1}(\bm{v}, Z) \wedge j \in Z} & = & \sum_{j \in L} p_j \probsub{\bm{v}, Z}{j \notin \Lambda_{n+1}(\bm{v}, Z)\ |\ j \in Z} \prob{j \in Z}\\
	&=& \frac{1}{2} \sum_{j \in L} p_j \probsub{\bm{v}, Z}{j \notin \Lambda_{n+1}(\bm{v}, Z)\ |\ j \in Z}.
\end{eqnarray*}
\qed
\end{proof}

We now prove \refLemma{lemma:ALG^B} using the above two propositions. Observe that the buyers' contribution to the social welfare $\ALG^B$ extracted by $\mathbb{M}_\text{1-supply}$ is equal to the sum of all the buyers' utilities and all the buyers' payments.

\begin{proof}[of \refLemma{lemma:ALG^B}]
As just observed above, from \refProposition{prop:utility_B} and \refProposition{prop:revenue_B}, we have that
\begin{eqnarray*}
\ALG^B &=& \mathbb{E}\left[\sum_{i \in [n]} u_i(\mathbb{M}_{\text{1-supply}}(\bm{v},\bm{w}))\right] + \sum_{j \in L} p_j \probsub{\bm{v}, Z}{j \notin \Lambda_{n+1}(\bm{v}, Z) \wedge j \in Z} \\
	&\geq& \frac{1}{2}\sum_{j \in L} \probsub{\bm{v}, Z}{j \in \Lambda_{n+1}(\bm{v}, Z)\ |\ j \in Z}p_j + \frac{1}{2} \sum_{j \in L} p_j \probsub{\bm{v}, Z}{j \notin \Lambda_{n+1}(\bm{v}, Z)\ |\ j \in Z}\\
	&=& \frac{1}{2} \sum_{j \in L} p_j = \frac{1}{4} \sum_{j \in L} \expected{\SW^B_j(\bm{v})}.
\end{eqnarray*}

By definition of $\bar{L}$, for each $j \in \bar{L}$ it holds that $4 \mathbb{E}[w_j] > \mathbb{E}[\SW^B_j(\bm{v})]$. Every item in $\bar{L}$ stays unsold so, 
\begin{equation*}
\ALG^S \geq \sum_{j \in \bar{L}} \expected{w_j} > \frac{1}{4} \sum_{j \in \bar{L}} \expected{\SW^B_j(\bm{v})}.
\end{equation*}
Therefore,
\begin{equation*}
\ALG^B + \ALG^S \geq \frac{1}{4} \sum_{j=1}^k \expected{\SW^B_j(\bm{v})}.
\end{equation*}

Now recall that $\mathbb{E}[\SW^B_j(\bm{v})]$ was defined by the allocation $\bm{X}^\all(\bm{v})$, being the one returned by Algorithm $\mathbb{A}$. So, 
\begin{equation*}
\frac{1}{4} \sum_{j=1}^k \expected{\SW^B_j(\bm{v})} = \frac{1}{4} \sum_{i=1}^n \expectedsub{\bm{v}}{v_i(X_i^\all(\bm{v}))} \geq \frac{1}{4\alpha} \OPT^B.
\end{equation*}
\qed
\end{proof}

\section{Proof of Theorem \ref{thm:bicmech} and Corollary \ref{cor:additive}}

\begin{proposition}
\label{prop:cover_sellers}
\begin{equation*}
\sum_{\ell \in L} \expected{\SW^B_\ell(\bm{v})} + 4\sum_{\ell \in \bar{L}} \expected{\SW^S_\ell(\bm{w})} > \sum_{i=1}^n \expected{v_i(X_i^\all(\bm{v}))}.
\end{equation*}
\end{proposition}

\begin{proof}
Let $a(v_i,X_i^\all(\bm{v}),\cdot)$ be the representative additive function of $v_i$ for the bundle $X_i^\all(\bm{v})$. Then,
\begin{eqnarray*}
\sum_{i=1}^n \expected{v_i(X_i^\all(\bm{v}))} &=& \sum_{i=1}^n \expected{\sum_{\ell \in X^\all_i(\bm{v})} a(v_i,X^\all_i(\bm{v}),\{\ell\})} \\
	& = & \sum_{i=1}^n \sum_{\ell=1}^k \expected{a(v_i, X^\all_i(\bm{v}),\{\ell\}) \indicator{\ell \in X^\all_i(\bm{v})}} \\
	& = & \sum_{\ell=1}^k \expected{\SW^B_\ell(\bm{v})} \\
	& = & \sum_{\ell \in L} \expected{\SW^B_\ell(\bm{v})} + \sum_{\ell \in \bar{L}} \expected{\SW^B_\ell(\bm{v})} \\
	& < & \sum_{\ell \in L} \expected{\SW^B_\ell(\bm{v})} + 4\sum_{\ell \in \bar{L}} \expected{\SW^S_\ell(\bm{w})}. 
\end{eqnarray*}
The last inequality follows because 
\begin{equation*}
4 \sum_{\ell \in \bar{L}} \expected{\SW^S_\ell(\bm{w})} > \sum_{\ell \in \bar{L}} \expected{\SW^B_\ell(\bm{v})},
\end{equation*}
by definition of $\bar{L}$.
\qed
\end{proof}

\begin{proposition}
\label{prop:utility_z}
Let $\bm{v}$ be a buyers' valuation function profile and let $(X'_1, \ldots, X'_n)$ be any allocation of items to the buyers, let $X'_{i, j} := X'_i \cap L_j$ be the set of items in $L$ that are allocated to buyer $i \in [n]$ and belonged to seller $j \in [m]$. For each seller $j \in [m]$, let $z_j \in \{0, 1\}$ be a Bernoulli random variable such that $\expected{z_j} = 1/2$. Let $X''_i(\bm{z}) := \bigcup_{j \in [m] : z_j = 1} X'_{i, j}$ for all $i \in [n]$. Then, for all $i \in [n]$ it holds that
\begin{equation*}
\expectedsub{\bm{z}}{v_i(X''_i(\bm{z}))} \geq \frac{1}{2} v_i(X'_i).
\end{equation*}
Moreover, given any vector $\bm{p} \in \mathbb{R}^k$ of item prices, the inequality also holds on the utilities of the buyers:
\begin{equation*}
\expectedsub{\bm{z}}{v_i(X''_i(\bm{z})) - \sum_{\ell \in X''_i(\bm{z})} p_\ell} \geq \frac{1}{2} \left( v_i(X_i') - \sum_{\ell \in X'_i} p_\ell \right).
\end{equation*}
\end{proposition}

\begin{proof}
For the first claim, first note that due to subadditivity
\begin{equation*}
\expectedsub{\bm{z}}{v_i(X''_i(\bm{z}))} \geq v_i(X'_i) - \expectedsub{\bm{z}}{v_i\left(\bigcup_{j \in [m] : z_j = 0} X'_{i, j}\right)}.
\end{equation*}
Observe that 
\begin{equation*}
\expectedsub{\bm{z}}{v_i(X''_i(\bm{z}))} = \expectedsub{\bm{z}}{v_i\left(\bigcup_{j \in [m] : z_j = 1} X'_{i, j}\right)} = \expectedsub{\bm{z}}{v_i\left(\bigcup_{j \in [m] : z_j = 0} X'_{i, j}\right)},
\end{equation*}
because the events $z_j = 0$ and $z_j = 1$ are equiprobable for all $j \in [m]$. Combining this with the above inequality establishes the first claim. 


The second claim follows from the following derivation.
\begin{eqnarray*}
\expectedsub{\bm{z}}{v_i(X''_i(\bm{z})) - \sum_{\ell \in X''_i(\bm{z})} p_\ell} & = & \expectedsub{\bm{z}}{v_i(X''_i(\bm{z}))} - \expectedsub{\bm{z}}{\sum_{j \in [m] : z_j = 1} \sum_{\ell \in X'_{i, j}} p_\ell} \\
	&=& \expectedsub{\bm{z}}{v_i(X''_i(\bm{z}))} - \expectedsub{\bm{z}}{\sum_{j = 1}^m \left(\sum_{\ell \in X'_{i, j}} p_\ell \right) \indicator{z_j = 1}}\\
	&=& \expectedsub{\bm{z}}{v_i(X''_i(\bm{z}))} - \sum_{j = 1}^m \left(\sum_{\ell \in X'_{i, j}} p_\ell \right)\expectedsub{\bm{z}}{\indicator{z_j = 1}}\\
	&=& \expectedsub{\bm{z}}{v_i(X''_i(\bm{z}))} - \sum_{\ell \in X'_i} p_\ell \frac{1}{2}\\
	&\geq& \frac{1}{2} v_i(X'_i) - \frac{1}{2} \sum_{\ell \in X'_i} p_\ell
\end{eqnarray*}
\qed
\end{proof}

\begin{proposition}
\label{prop:sellers_1/2}
Let $j \in [m]$ be a seller. The probability that the mechanism $\mathbb{M}_{\text{add}}$ makes in Step \ref{offer} an offer to $j$ that she accepts, is $1/2$.
\end{proposition}

\begin{proof}
For every $j \in [m]$ and $\ell \in L_j$, it holds by definition of $p_\ell$ and $L_j$ that $p_\ell \geq 2\expected{w_j(\{\ell\})}$. From Markov's inequality it follows that
\begin{equation*}
\prob{w_j(S_j) > \sum_{\ell \in S_j} p_\ell} \leq \prob{w_j(S_j) > 2 \expected{w_j(S_j)}} < \frac{1}{2}.
\end{equation*}
Thus, $\prob{w_j(S_j) \leq \sum_{\ell \in S_j} p_\ell} \geq 1/2$, meaning that $j$ accepts the offer with probability at least $1/2$, in case she is made an offer. The mechanism makes the offer with probability $q_j$, and 
\begin{equation*}
q_j \prob{w_j(S_j) \leq \sum_{\ell \in S_j} p_\ell} = 1/2.
\end{equation*}
\qed
\end{proof}

For $i \in [n+1]$ and valuation profile $\bm{v}$, let $\bm{v}_{<i} = (v_1, \ldots, v_{i-1})$ and let $\Lambda_i(\bm{v}_{<i})$ be the set $\Lambda_i$ defined in Step \ref{lambda}, when $\mathbb{M}_{\text{add}}$ is run when the buyers in $[i-1]$ have valuation profile $\bm{v}_{<i}$. Given this definition, the set $\Lambda_{n+1}(\bm{v})$ are the items not requested by any buyer at the end of Step \ref{buyerloop}, when the buyers' valuation profile is $\bm{v}$.
\begin{lemma}
The expected total utility of the buyers is at least
\begin{equation*}
\frac{1}{2} \sum_{\ell \in L} \mathsf{Pr}_{\bm{v}}[\ell \in \Lambda_{n+1}(\bm{v})] p_\ell.
\end{equation*}
\end{lemma}
\begin{proof}
First, let us consider a fixed buyer $i \in [n]$ and a fixed buyers' valuation profile $\bm{v}$. Let $\tilde{\bm{v}}_{-i}$ be an independently sampled valuation profile for the buyers in $[n]\setminus\{i\}$, and consider the bundle $X_i^\all(v_i, \tilde{\bm{v}}_{-i})$ that $\mathbb{A}$ allocates to $i$ when the valuation profile is $(v_i, \tilde{\bm{v}}_{-i})$. Let $X_i^L(\bm{v}, \tilde{\bm{v}}_{-i}) = X_i^\all(v_i, \tilde{\bm{v}}_{-i}) \cap L \cap \Lambda_i(\bm{v}_{<i})$. 
Moreover, let $\bm{z}$ be a vector of $m$ Bernoulli random variables with $\expected{z_j} = 1/2$ and define for a subset $S(\bm{v}) \subseteq \Lambda_i{\bm{v}}$ the random variable $S(\bm{v},\bm{z}) = \bigcup_{j \in [m] : z_j = 1} (S \cap L_j)$. Particularly, from this definition we obtain the random variable $X_i(\bm{v}, \tilde{\bm{v}}_{-i}, \bm{z}) = \bigcup_{j \in [m] : z_j = 1} (X^L_{i}(\bm{v}, \tilde{\bm{v}}_{-i}) \cap L_j)$. Also, note that when the buyers' valuations are $\bm{v}$, the mechanism will let $i$ choose to request a bundle from the set $\Lambda_i(\bm{v}_{<i})$ with item prices $\bm{p}$, the buyer maximizes her expected utility and will therefore request the bundle $B_i(\bm{v})$ that maximizes her expected utility which can be expressed as 
\begin{equation*}
\mathbb{E}_{\bm{z}}\left[v_i(B(\bm{v},\bm{z})) - \sum_{\ell \in B(\bm{v},\bm{z})} p_{\ell}\right],
\end{equation*}
as by Propostion \ref{prop:sellers_1/2} each sellers' requested items will be allocated with probability $1/2$, as reflected by the Bernouilli variables $\bm{z}$.

Since $B_i(\bm{v})$ is an expected-utility-maximising bundle $X_{i, j}(\bm{v}, \tilde{\bm{v}}_{-i}, \bm{z}) \subseteq \Lambda_i(\rm{v})$, it holds that 
\begin{eqnarray*}
\mathbb{E}_{\bm{z}}\left[v_i(B(\bm{v},\bm{z})) - \sum_{\ell \in B(\bm{v},\bm{z})} p_{\ell}\right] & \geq & \expectedsub{\tilde{\bm{v}}_{-i}, \bm{z}}{v_i(X_i(\bm{v}, \tilde{\bm{v}}_{-i}, \bm{z})) - \sum_{\ell \in X_i(\bm{v}, \tilde{\bm{v}}_{-i}, \bm{z})} p_{\ell}} \\
 & \geq & \frac{1}{2} \expectedsub{\tilde{\bm{v}}_{-i}}{v_i(X^L_i(\bm{v}, \tilde{\bm{v}}_{-i})) - \sum_{\ell \in X^L_i(\bm{v}, \tilde{\bm{v}}_{-i})} p_{\ell}} \\
 & \geq & \frac{1}{2} \expectedsub{\tilde{\bm{v}}_{-i}}{a(v_i,X^\all_i(v_i, \bm{v}_{-i}), X_i^L(\bm{v}, \bm{v}_{-i})) - \sum_{\ell \in X^L_i(\bm{v}, \tilde{\bm{v}}_{-i})} p_\ell}\\
	&=& \frac{1}{2} \expectedsub{\tilde{\bm{v}}_{-i}}{\sum_{\ell \in X^L_i(\bm{v}, \tilde{\bm{v}}_{-i})} (\SW^B_\ell(v_i, \tilde{\bm{v}}_{-i}) - p_\ell)}.
\end{eqnarray*}
where the second inequality follows from \refProposition{prop:utility_z}, and the last inequality follows from the definition of the additive representative function $a(v_i,X^\all_i(v_i, \tilde{\bm{v}}_{-i}, \cdot)$. 

If we sum over all $i \in [n]$ and take the expectation w.r.t. every $v_i$, we obtain the following bound on the total expected utility of the buyers.
\begin{align*}
& \mathbb{E}_{\bm{v},\bm{z}}\left[\sum_{i = 1}^n (v_i(B(\bm{v},\bm{z})) - \sum_{\ell \in B(\bm{v},\bm{z})} p_{\ell})\right] \geq \frac{1}{2} \expectedsub{\bm{v},\tilde{\bm{v}}_{-i}}{\sum_{i = 1}^n \sum_{\ell \in X^L_i(\bm{v}, \tilde{\bm{v}}_{-i})} (\SW^B_\ell(v_i, \tilde{\bm{v}}_{-i}) - p_\ell)} \\
& \qquad = \frac{1}{2} \expectedsub{\bm{v}, \tilde{\bm{v}}_{-i}}{\sum_{i=1}^n \sum_{\ell \in L} (\SW^B_\ell(v_i, \tilde{\bm{v}}_{-i}) - p_\ell) \indicator{\ell \in X^L_i(\bm{v}, \tilde{\bm{v}}_{-i})}}\\
& \qquad = \frac{1}{2} \expectedsub{\bm{v}, \tilde{\bm{v}}_{-i}}{\sum_{i=1}^n \sum_{\ell \in L} (\SW^B_\ell(v_i, \tilde{\bm{v}}_{-i}) - p_\ell) \indicator{\ell \in X_i^\all(v_i, \tilde{\bm{v}}_{-i})} \indicator{\ell \in \Lambda_i(\bm{v}_{<i})}} \\
& \qquad = \frac{1}{2} \sum_{\ell \in L} \sum_{i=1}^n \expectedsub{v_i, \tilde{\bm{v}}_{-i}}{(\SW^B_\ell(v_i, \tilde{\bm{v}}_{-i}) - p_\ell) \indicator{\ell \in X_i^\all(v_i, \tilde{\bm{v}}_{-i})}} \expectedsub{\bm{v}_{-i}}{\indicator{\ell \in \Lambda_i(\bm{v}_{<i})}}.
\end{align*}
For the second-to-last equality, we exploited the independence of the events $(\ell \in X_i^\all(v_i, \tilde{\bm{v}}_{-i}))$ and $(\ell \in \Lambda_i(\bm{v}_{<i}))$. 
Then, $\expectedsub{\bm{v}_{-i}}{\indicator{\ell \in \Lambda_i(\bm{v}_{<i})}} = \prob{\ell \in \Lambda_i(\bm{v}_{<i})}$ and since $L = \Lambda_1(\bm{v}_{<1}) \supseteq \ldots \supseteq \Lambda_{n+1}(\bm{v})$, it holds that $\prob{\ell \in \Lambda_i(\bm{v}_{<i})} \geq \prob{\ell \in \Lambda_{n+1}(\bm{v})}$. 
So, we have that the above expression is at least
\begin{eqnarray*}
	&& \frac{1}{2} \sum_{\ell \in L} \probsub{\bm{v}}{\ell \in \Lambda_{n+1}(\bm{v})} \sum_{i=1}^n \expectedsub{v_i, \tilde{\bm{v}}_{-i}}{(\SW^B_\ell(v_i, \tilde{\bm{v}}_{-i}) - p_\ell) \indicator{\ell \in X_i^\all(v_i, \tilde{\bm{v}}_{-i})}}\\
	&=& \frac{1}{2} \sum_{\ell \in L} \probsub{\bm{v}}{\ell \in \Lambda_{n+1}(\bm{v})} \sum_{i=1}^n \expectedsub{\bm{v}}{(\SW^B_\ell(\bm{v}) - p_\ell) \indicator{\ell \in X_i^\all(\bm{v})}}. 
\end{eqnarray*}
The equality follows from renaming the random variable $v_j := \tilde{v}_j$ for all $j \neq i$ 
 Now observe that by definition of the prices, $p_\ell = \sum_{i=1}^n \expectedsub{\bm{v}}{(\SW^B_\ell(\bm{v}) - p_\ell) \indicator{\ell \in X_i^\all(\bm{v})}}$. Combining these derivations, we obtain the desired bound on the expected utilities
\begin{equation*}
\mathbb{E}_{\bm{v},\bm{z}}\left[\sum_{i = 1}^n (v_i(B(\bm{v},\bm{z})) - \sum_{\ell \in B(\bm{v},\bm{z})} p_{\ell})\right] \geq \frac{1}{2} \sum_{\ell \in L} \mathsf{Pr}_{\bm{v}}[\ell \in \Lambda_{n+1}(\bm{v})] p_\ell.
\end{equation*}
\qed
\end{proof}

\begin{lemma}
The expected sum of payments made by the buyers is equal to
\begin{equation*}
\frac{1}{2} \sum_{\ell \in L} \probsub{\bm{v}}{\ell \notin \Lambda_{n+1}(\bm{v})} p_\ell.
\end{equation*}
\end{lemma}

\begin{proof}
For $j \in [m]$, let $z_j$ be the random $(0,1)$-variable that indicates whether seller $j$ has been made an offer and accepted it in Step \ref{offer} of Mechanism $\mathbb{M}_{\text{add}}$, so $z_j = 1$ is a Bernouilli variable with expected value $1/2$. The expected sum of payments made by the buyers is then
\begin{eqnarray*}
\sum_{j = 1}^m \sum_{\ell \in L_j} \prob{\ell \notin \Lambda_{n+1}(\bm{v}) \wedge z_j = 1} p_\ell & = & \sum_{j = 1}^m \sum_{\ell \in L_j} \prob{\ell \notin \Lambda_{n+1}(\bm{v})} \prob{z_j = 1} p_\ell \\
 & = & \frac{1}{2} \sum_{j = 1}^m \sum_{\ell \in L_j} \prob{\ell \notin \Lambda_{n+1}(\bm{v})} p_\ell\\
 & = & \frac{1}{2} \sum_{\ell \in L} \prob{\ell \notin \Lambda_{n+1}(\bm{v})} p_\ell
\end{eqnarray*}
The second equality holds by the independence of the two events.
\qed
\end{proof}

\begin{lemma}
\begin{equation*}
\ALG^B \geq \frac{1}{4} \sum_{\ell \in L} \expectedsub{\bm{v}}{\SW^B_\ell(\bm{v})}.
\end{equation*}
\end{lemma}

\begin{proof}
The expected social welfare contribution of the buyers is equal to the sum of the expected utilities and expected payments. 
By the above two lemmas, their sum is at least
\begin{equation*}
\frac{1}{2} \sum_{\ell \in L} \prob{\ell \in \Lambda_{n+1}(\bm{v})} p_\ell + \frac{1}{2} \sum_{\ell \in L} \prob{\ell \notin \Lambda_{n+1}(\bm{v})} p_\ell = \frac{1}{2} \sum_{\ell \in L} p_\ell = \frac{1}{4} \sum_{\ell \in L} \expectedsub{\bm{v}}{\SW^B_\ell(\bm{v})},
\end{equation*}
by definition of $p_\ell$.
\qed
\end{proof}

\begin{lemma}
\begin{equation*}
4\alpha \ALG^B + 4\alpha \ALG^S \geq \OPT^B.
\end{equation*}
\end{lemma}
\begin{proof}
By the above lemma, $4 \ALG^B \geq \sum_{\ell \in L} \mathbb{E}_{\bm{v}}[\SW^B_\ell(\bm{v})]$. Moreover, our mechanism leaves every item $\ell \in \bar{L}$ with its seller, and so $4\ALG^S \geq 4 \sum_{\ell \in \bar{L}} \mathbb{E}_{\bm{w}}[\SW^S_\ell(\bm{v})]$. Therefore,
\begin{equation*}
4 \ALG^B + 4 \ALG^S \geq \sum_{\ell \in L} \mathbb{E}_{\bm{v}}{\SW^B_\ell(\bm{v})} + 4 \sum_{\ell \in \bar{L}} \mathbb{E}_{\bm{w}}[\SW^S_\ell(\bm{v})] \geq \sum_{i=1}^n \mathbb{E}_{\bm{v}}{v_i(X_i^\all(\bm{v}))} \geq \frac{1}{\alpha}\OPT^B,
\end{equation*}
The second inequality holds by \refProposition{prop:cover_sellers}, and the last inequality follows because we defined $\alpha$ to be the approximation factor of algorithm $\mathbb{A}$, which is the algorithm that we assumed to generate allocation $X^\all(\bm{v})$.

\qed
\end{proof}

\begin{lemma}
\begin{equation*}
2 \ALG^S \geq \OPT^S.
\end{equation*}
\end{lemma}

\begin{proof}
The only items that our mechanisms potentially reallocates are the ones belonging to $L$. Every item in $\bar{L}$ stays with its seller. For the items in $L$, the mechanism ensures every seller sells her demanded bundle with probability exactly $1/2$, so for each seller it holds that she retains her full initial endowment with probability at least $1/2$, which implies the claim.
\qed
\end{proof}

\begin{proof}[of Theorem \ref{thm:bicmech}]
Every buyer chooses a bundle that maximizes her expected utility, so the mechanism is ex-interim IR and BIC on the buyers' side. On the sellers' side, it is actually ex-post IR and DSIC: the sellers solely have to decide between accepting or rejecting a single offer to receive a proposed payment in exchange for a bundle of items, and it is clearly a dominant strategy to accept if and only if such an exchange leads to an improvement in the seller's utility. The fact that the mechanism is DSBB follows from its definition, which makes clear that payments are defined by the appropriate sequence of trades and payments from buyers to sellers. The approximation guarantee follows by the sum of the inequalities of the above two lemmas.
\qed
\end{proof}
By taking for $\mathbb{A}$ an optimal algorithm (i.e., $\alpha = 1$), we obtain the existence of a mechanism that is ex-post IR, DSIC, DSBB, and $6$-approximates the optimal social welfare.
Again, one may also take for $\mathbb{A}$ a polynomial time approximation algorithm in order to obtain a polynomial time mechanism.

\begin{proof}[of Corollary \ref{cor:additive}]
If a buyer $i \in [n]$ has an additive valuation function, it is a dominant strategy to request the items in $\Lambda_i(\bm{v}_{<i})$) for which it holds that $v_i(\{\ell\}) > p_\ell$. This follows from the simple fact that by additivity, the utility that a player has for any bundle of items $S$ can be written as 
$\sum_{\ell \in S} v_i(\{\ell\}) - p_{\ell}$.
Thus, for every item $\ell \in [k]$ that a buyer requests (recall that this item is then allocated to her for price $p_\ell$ with probability $1/2$), a term of $(1/2) (v_i(\{\ell\}) - p_{\ell})$ gets added to her expected utility. So including $\ell$ in her requested bundle is profitable if and only if $v_i(\{\ell\}) - p_{\ell} \geq 0$. Ex-post IR property is also satisfied by following this strategy.
\qed
\end{proof}

\end{document}